\documentclass[twocolumn,pra,longbibliography,aps,10pt,floatfix]{revtex4-2}
\usepackage[utf8]{inputenc}
\usepackage[english]{babel}
\usepackage{fontenc}
\usepackage{anyfontsize}
\usepackage[toc,page]{appendix}
\usepackage{amsmath}
\usepackage{amssymb}
\usepackage[mathlines]{lineno}
\usepackage{algorithm}
\usepackage{algpseudocode}
\usepackage{algcompatible}
\usepackage{graphicx}
\usepackage{wrapfig}
\usepackage{newtxtext,newtxmath}
\usepackage[colorlinks = true, linkcolor = blue, urlcolor=blue, citecolor = red, anchorcolor = green]{hyperref}
\usepackage{comment}
\usepackage{mathtools}

\floatname{algorithm}{Protocol}
\newlength{\continueindent}
\usepackage{etoolbox}
\makeatletter
\newcommand*{\ALG@customparshape}{\parshape 2 \leftmargin \linewidth \dimexpr\ALG@tlm+\continueindent\relax \dimexpr\linewidth+\leftmargin-\ALG@tlm-\continueindent\relax}
\apptocmd{\ALG@beginblock}{\ALG@customparshape}{}{\errmessage{failed to patch}}
\makeatother

\DeclarePairedDelimiter\floor{\lfloor}{\rfloor}

\newtheorem{theorem}{Theorem}

\newtheorem{lemma}{Lemma}

\newtheorem{definition}{Definition}

\newtheorem{proposition}{Proposition}
\newenvironment{proof}[1][Proof]{\noindent\textbf{#1.} }{\ \rule{0.5em}{0.5em}}

\allowdisplaybreaks

\begin{document}
\preprint{APS/123-QED}

\title{Device-Independent Certification of Multipartite Distillable Entanglement}

\author{Aby Philip}
\affiliation{School of Applied and Engineering Physics,
Cornell University, Ithaca, New York 14850, USA}
\author{Mark M. Wilde}
\affiliation{School of Electrical and Computer Engineering and School of Applied and Engineering Physics,
Cornell University, Ithaca, New York 14850, USA}

\begin{abstract}
    Quantum networks consist of various quantum technologies, spread across vast distances, and involve various users at the same time. Certifying the functioning and efficiency of the individual components is a task that is well studied and widely used. However, the power of quantum networks can only be realized by integrating all the required quantum technologies and platforms across a large number of users. In this work, we demonstrate how to certify the distillable  entanglement available in multipartite states produced by quantum networks, without relying on the physical realization of its constituent components. We do so by using the paradigm of device independence.
\end{abstract}

\maketitle

\tableofcontents

\section{Introduction}

Quantum networks may become a reality within the near future and have the potential to open up several avenues of applications. Potential uses for quantum networks include quantum computers connected together for distributed computing tasks~\cite{dhondt2006computational,PhysRevLett.86.5188}, a collection of quantum sensors that implement a joint measurement on a system of interest~\cite{PhysRevA.86.014301,PhysRevA.97.042337,PhysRevA.100.022316,PhysRevA.103.L030601}, or a number of distant nodes that transmit quantum states among themselves~\cite{PhysRevA.59.1829,zhu2015w,MurtaQCKA}. To realize the full potential of quantum networks, the efficient production and distribution of multipartite entanglement is essential~\cite{PhysRevA.98.052313}.

To distribute multipartite entanglement over large distances, one can choose two basic approaches. The first approach involves distributing bipartite entanglement between each of the nodes and then using local operations and classical communication to convert the global state into a desired multipartite state~\cite{PhysRevLett.125.080502,PhysRevLett.109.150403}. The alternate method involves producing the multipartite state of interest at a single location~\cite{PhysRevLett.132.130604} and distributing the resulting state via quantum channels to the intended recipients. Whichever procedure one might employ, it is important to study the associated success probability and quality of entanglement yielded by these procedures. 

This point raises the question: how can we certify or guarantee that a minimum amount of entanglement is produced by a certain procedure or method? This question becomes more interesting when we recall that quantum networks may have components based on different architectures: for example, superconducting qubits for computing~\cite{PhysRevResearch.2.013004,zhong2021deterministic,PhysRevA.91.012333}, trapped ions~\cite{PhysRevLett.124.110501} or solid-state qubits~\cite{multimode_Q_network} for memory, photons for communication~\cite{PhysRevLett.132.130604}, and so on. Hence, we need to adopt a framework of certification that is independent of the constituent parts of the quantum network, thus bringing us to device-independent (DI) protocols. 

The general interest in DI protocols~\cite{Primaatmaja2023securityofdevice,zapatero2023advances} stems from the fact that one can test and verify the amount of entanglement in a state using only classical inputs and outputs, along with classical communication. DI protocols are based on various self-testing results~\cite{MayersYao04,PhysRevA.90.042340,matthew2014selftest,PhysRevLett.117.070402}, the first of which were developed in~\cite{MayersYao04}. (See \cite{Supic2020selftestingof} for a review.) Most importantly, such protocols are agnostic to the underlying mechanics and specifics of the state being tested. As such, and relevant to the present paper, DI protocols are suitable for the certification of multipartite entanglement in quantum networks, especially ones that contain elements with dissimilar underlying technologies. 

There has been a lot of recent interest in device-independent protocols involving more than two parties, including investigations into Bell-type inequalities~\cite{PhysRevA.109.022205,Le2023quantumcorrelations} and their application toward DI conference key agreement, involving achievable rates~\cite{ribeiro2018fully,holz2020genuine} and upper bounds on key-agreement rates~\cite{Philip2023multipartite,PhysRevA.105.022604}. However, little attention has been paid toward the device-independent certification of multipartite distillable entanglement. Multipartite entanglement distillation is an important step to ensure that we are using states that are as close as possible to ideal states. 

In this paper, we provide a lower bound on the certifiable rate of one-shot multipartite distillable entanglement in a device-independent setting. The multipartite Greenberger–Horne–Zeilinger (GHZ) state is the basic entanglement resource that is critical to maximizing the utility of quantum networks for applications like quantum sensing~\cite{PhysRevA.86.014301,PhysRevA.97.042337,PhysRevA.100.022316,PhysRevA.103.L030601}, multi-party quantum communication~\cite{PhysRevA.59.1829,zhu2015w,MurtaQCKA}, and distributed quantum computation~\cite{dhondt2006computational,PhysRevLett.86.5188}. It is vital to distill GHZ states from less-than-perfect multipartite states that are available to the parties in a quantum network, while using only minimal additional resources like local operations and classical communication (LOCC). To maximize the utility of every copy of an available  multipartite state, we need to look at the one-shot distillable entanglement of the state shared by the quantum network. The device-independent setting means that all the conclusions are drawn from only the classical inputs and outputs of the protocol, along with the assumption of the completeness and correctness of quantum mechanics~\cite{MayersYao04,PhysRevA.90.042340,matthew2014selftest}. In other words, the conclusions hold true regardless of the precise inner workings of the experimental devices.

In more detail, we extend DI entanglement distillation certification from the bipartite scenario~\cite{Arnon_Friedman_2019} to the multipartite scenario. In what follows, we first give a detailed description and definition of what we mean by DI multipartite entanglement distillation. In Section~\ref{sec:def_n_set}, we define a DI, $M$-partite entanglement distillation certification protocol, which consists of completeness and soundness conditions. In Section~\ref{sec:protocol}, we provide a detailed description of a proposed protocol. Our proposed protocol is centered around the Mermin--Ardehali--Belinskii--Klyshko (MABK) inequality~\cite{PhysRevA.46.5375,PhysRevLett.65.1838,AVBelinskiĭ_1993}, which is critical to proving the completeness condition. We show in Section~\ref{sec:completeness} that the proposed protocol is complete. Thereafter, in Section~\ref{sec:soundness}, we proceed to proving that it is sound, by making use of the entropy accumulation theorem~\cite{dupuis2016entropy} and the structure of the MABK~inequality~\cite{ribeiro2017MABK}.

\section{Definitions and Setup}\label{sec:def_n_set}

In this section, we provide several definitions that we use throughout the rest of the paper, including the GHZ orthonormal basis, the one-shot multipartite distillable entanglement of a multipartite quantum state, an entanglement distillation certification protocol, and basic entropies.

\begin{definition}\label{def:ghz_basis}
    The GHZ orthonormal basis for the set of $M$-qubit states is composed of the following $2^M$ states:
    \begin{equation}\label{eqn:def_ghz_basis}
        |\psi_{v,u}\rangle \coloneqq \frac{1}{\sqrt{2}}\left[|0,u\rangle +(-1)^v |1,\bar{u}\rangle\right],
    \end{equation}
    where $v\in \{0,1\}$, while $u\in\{0,1\}^{M-1}$ and $\bar{u}=\boldsymbol{1}\oplus u$ are $M-1$ bit strings, with $\boldsymbol{1}$ the all-ones bit vector of size $M-1$. 
\end{definition}

The one-shot multipartite distillable entanglement of a state quantifies the amount of multipartite entanglement that we can distill from a single copy of the state. It is defined formally as follows:

\begin{definition}
[Multipartite distillable entanglement]
\label{def:one-shot-distill} Let $n\in\mathbb{N}$ and $\varepsilon\in\lbrack0,1]$. The one-shot distillable $M$-partite entanglement $E_{D}^{\varepsilon}(\rho_{A_{[M]}})$ of a multipartite state $\rho_{A_{[M]}}\equiv\rho_{A_{1}\cdots A_{M}}$ is defined as
    \begin{multline}\label{eqn:one_shot_distillable_ent}
    E_{D}^{\varepsilon}(\rho_{A_{[M]}}) \coloneqq \\ 
        \sup_{\substack{d \in \mathbb{N},\\ \mathcal{L}\in\mathrm{LOCC}}}\left\{  \log_{2}{d}: F(\mathcal{L}_{A_{[M]}\rightarrow\hat{A}_{[M]}}(\rho),\Phi_{\hat{A}_{[M]}}^{d})\geq1-\varepsilon\right\},
    \end{multline}
where the optimization is over every LOCC channel $\mathcal{L}_{A_{[M]}\rightarrow\hat{A}_{[M]}}$, the fidelity is defined as $F(\rho,\sigma)\coloneqq \left\Vert\sqrt{\rho}\sqrt{\sigma}\right\Vert_1^2$, and $\Phi_{\hat{A}_{[M]}}^{d}\equiv |\Phi^{d}\rangle\!\langle\Phi^{d}|_{\hat{A}_{[M]}}$ is an $M$-party, rank-$d$ GHZ state, with%
    \begin{equation}
        |\Phi_{\hat{A}_{[M]}}^{d}\rangle\coloneqq \frac{1}{\sqrt{d}}\sum_{i=0}^{d-1}|{i}\rangle_{A_{1}}\cdots|{i}\rangle_{A_{M}}.
    \end{equation}
The asymptotic distillable entanglement is defined in terms of the following limit:
    \begin{equation}
        E_{D}(\rho_{A_{[M]}}) \coloneqq \inf_{\varepsilon\in(0,1)}\liminf_{n\rightarrow\infty}\frac{1}{n}E_{D}^{\varepsilon}(\rho_{A_{[M]}}^{\otimes n}).
    \end{equation}
\end{definition}

Since our goal is to quantify multipartite distillable entanglement in a device-independent setting, we want to be able to make statements about the aforementioned quantity that hold regardless of the physical systems involved. In other words, we want to bound this quantity in a device-independent setting.

Let us take a closer look at what this means. To begin with, suppose that all parties $A_{1},\ldots,A_{M}$ are given a share of a multipartite quantum state $\rho_{A_{1}\cdots A_{M}}$, which is distributed to them by a possibly unknown entity. Each party also has access to a black box with which they can interact classically. For each classical input, the corresponding black box applies a positive operator-valued measure (POVM) on its respective share of the multipartite state. After the application of the POVM, the box outputs a classical value that is recorded by the corresponding participant. The parties can use the results of the measurements to complete certain tasks of their choosing. If they use only the inputs and outputs from the black box to complete their task, they will have completed the task independent of the underlying physical realization of the measurement and states shared among them. This is the idea behind device independence. 

We now define what we mean by a device-independent (DI) protocol for certifying a rate of one-shot distillable multipartite entanglement. Our definition for a DI $M$-partite entanglement distillation certification (DIMEC) protocol is based on the definition of a DI entanglement certification protocol in \cite{Arnon_Friedman_2019}, defined therein for two parties.

\begin{definition}[DIMEC protocol]\label{def:dimec_protocol}
    Let $n\in\mathbb{N}$, let $\varepsilon_{\mathrm{smo}}$, $\varepsilon_{\operatorname{snd}}$, $\varepsilon_{\mathrm{cmp}}\in[0,1]$, and let $r \geq 0$ be a threshold $M$-partite distillation rate. Furthermore, let $\mathcal{S}^{\mathrm{honest}}$ be a set of ``honest'' states, each of which is denoted by $\phi^{\mathrm{honest}}$. Let $\mathcal{D}^{\mathrm{honest}}$ be the set of  ``honest'' measurement devices. Let $\mathcal{P}$ be a protocol employing only multipartite LOCC, which, upon being given a state $\sigma\in\mathcal{D}\!\left(\bigotimes_{i=1}^{M}\mathcal{H}_{A_{i}}^{\otimes n}\right)$, creates a state $\rho\in\mathcal{D}\!\left(\bigotimes_{i=1}^{M}\mathcal{H}_{A_{i}}^{\otimes n}\right)$. Let $\rho_{|\Omega}$ denote the final state conditioned on the protocol not aborting.
    
    A protocol $\mathcal{P}$ is said to be a DI $M$-partite entanglement distillation certification (DIMEC) protocol if the following conditions hold:

    \begin{enumerate}
        \item Noise tolerance (completeness): The probability that $\mathcal{P}$ aborts when applied on $\phi^{\operatorname{honest}}\in\mathcal{S}^{\operatorname{honest}}$ using a measurement device from $\mathcal{D}^{\operatorname{honest}}$ is at most~$\varepsilon_{\operatorname{cmp}}$.\label{cond:completeness}

        \item Entanglement certification (soundness): For every source $\sigma$ and measurement device, either $E_{D}^{n,\varepsilon_{\operatorname{snd}}}(\rho_{|\Omega})\geq r$ or $\mathcal{P}$ aborts with probability greater than $1-\varepsilon_{\operatorname{smo}}$ when applied on $\sigma$.\label{cond:soundness}
    \end{enumerate}
\end{definition}

To design a protocol and show that it satisfies all the conditions in Definition~\ref{def:dimec_protocol}, we need to use some information-theoretic quantities. Here we recall the definitions of von Neumann entropy, coherent information, and smooth conditional max-entropy. 
\begin{definition}
    The von Neumann entropy $H(A)_{\rho}$ of a state $\rho_{A}$ is defined as follows:
    \begin{equation}
        H(A)_{\rho} \coloneqq -\operatorname{Tr}[\rho_{A}\log_2 \rho_{A}].
    \end{equation}
\end{definition}

\begin{definition}
    The coherent information $I(A\rangle B)_{\rho}$ of a bipartite state $\rho_{AB}$ is defined as follows:
    \begin{equation}
        I(A\rangle B)_{\rho} \coloneqq H(B)_{\rho} -H(AB)_{\rho} =  -H(A|B)_{\rho},
    \end{equation}
    where $H(B)_{\rho}$ is the von Neumann entropy of $\rho_{B}=\operatorname{Tr}_{A}[\rho_{AB}]$ and $H(A|B)_{\rho} \coloneqq H(AB)_{\rho} - H(B)_{\rho} $ is the quantum conditional entropy of the state $\rho_{AB}$.
\end{definition}

\begin{definition}[\cite{Tomamichel17}]
    The smooth conditional max-entropy of a bipartite state $\rho_{AB}$ is defined for all $\varepsilon\in (0,1)$ as 
    \begin{equation}
        H^{\varepsilon}_{\max}(A\vert B)_{\rho}\coloneqq \inf_{\tilde{\rho} \in \mathcal{D}_{\leq}}H_{\max}(A\vert B)_{\rho},
    \end{equation}
     where
     \begin{equation}    
     \mathcal{D}_{\leq} \coloneqq \{\omega \geq 0 : \operatorname{Tr}[\omega]\leq 1\},
     \end{equation}
     $\tilde{\rho}$ is such that $\sqrt{1-F(\rho,\tilde{\rho})}\leq \varepsilon$, $H_{\max}(A\vert B)_{\rho}\coloneqq  \sup_{\sigma_{B}}\log_2 F(\rho_{AB},\mathbb{I}_A\otimes\sigma_{B})$, and $F(\rho,\sigma)\coloneqq \left\Vert \sqrt{\rho} \sqrt{\sigma} \right\Vert_1^2$ is the fidelity.
\end{definition}

\section{Device-Independent Multipartite Entanglement Certification Protocol}\label{sec:protocol}

    Before introducing our proposed DIMEC protocol, first we need to discuss the MABK inequality~\cite{PhysRevA.46.5375,PhysRevLett.65.1838,AVBelinskiĭ_1993}. Let~$x_{i}$ denote an input to the $i$th measurement device, and let $a_{i}$ denote the outcome of a measurement, where $i\in\{1,\ldots,M\}$ and $M$ is the number of parties involved. We can then define the MABK inequality as follows. 

    \begin{definition}\label{def:MABK}
        Let $\hat{O}_0^i$ and $\hat{O}_1^i$ be binary observables for all $ i \in [M]$. The $M$-partite MABK operator, $\mathcal{K}_M$, is defined by the following recursion relation:
        \begin{align}\label{eqn:def_mabk_operator}
            \mathcal{K}_M &\coloneqq  \frac{1}{2}\mathcal{F}(\mathcal{K}_{M-1},\overline{\mathcal{K}}_{M-1},\hat{O}_0^M,\hat{O}_1^M),\\
            \label{eqn:def_mabk_operator-2}
            \mathcal{K}_2 &\coloneqq  \frac{1}{2}\mathcal{F}(\hat{O}_0^1,\hat{O}_1^1,\hat{O}_0^2,\hat{O}_1^2),
    \end{align}
    where $\overline{\mathcal{K}}_{M-1}$ is obtained from $\mathcal{K}_{M-1}$ by exchanging  $\hat{O}_0^i$ and $\hat{O}_1^i$ for all $ i\in[M]$ and 
    \begin{multline}\label{eqn:def_bell}
        \mathcal{F}(\hat{B}_0, \hat{B}_1, \hat{C}_0, \hat{C}_1) \\
        \coloneqq \hat{B}_0\otimes \left(\hat{C}_0 + \hat{C}_1\right) + \hat{B}_1\otimes\left(\hat{C}_0 - \hat{C}_1\right).
    \end{multline}
    The $M$-partite MABK inequalities are then defined for all $ M\geq 2$ as
    \begin{equation}\label{eqn:inequality_mabk}
        2^{\frac{4-M}{2}}\left|\operatorname{Tr}\!\left[\mathcal{K}_M\, \rho_{\hat{A}_{[M]}}\right]\right|\leq 2^{\frac{m-M+3}{2}},\quad m\in[M].
    \end{equation}
\end{definition}

    The MABK inequalities are such that a violation of the inequalities for $m = 1$ proves that at least two parties are entangled, the violation of the inequalities for $m = M-1$ proves genuine $M$-partite entanglement, and the case where $m = M$ gives an upper bound (tight) on what is achievable by quantum mechanics. In this work, we are interested in $m=M-1$ and $m=M$, which correspond to 
    \begin{equation}\label{eqn:inequality_mabk_limits}
        \beta_M\coloneqq 2^{\frac{4-M}{2}}\left|\operatorname{Tr}\!\left[\mathcal{K}_M\, \rho_{\hat{A}_{[M]}}\right]\right|\in [2,2\sqrt{2}].
    \end{equation}
    The CHSH inequality corresponds to
    \begin{equation}\label{eqn:inequality_chsh_limits}
        \beta_2\coloneqq 2\left|\operatorname{Tr}\!\left[\mathcal{K}_2\, \rho_{\hat{A}_{[2]}}\right]\right|\in [2,2\sqrt{2}].
    \end{equation}

    For our purposes, we need to turn the MABK inequality into a game. This can be done using the procedure outlined in~\cite{1313847}. By unraveling the recursion in~\eqref{eqn:def_mabk_operator}--\eqref{eqn:def_mabk_operator-2}, we can rewrite the $M$-MABK operator as
    \begin{equation}
    \label{eqn:mabk_operator_rewrite}
        \mathcal{K}_M = 2^{-2\floor*{\frac{M}{2}}}\sum_{x\in\{0,1\}^M}(-1)^{f(x)}\bigotimes_{i\in[M]}\hat{O}_{x_i}^i ,
    \end{equation}
    where $x_i\in\{0,1\}$ is the $i$th bit of $x$ and $f:\left\{0,1\right\}^M\rightarrow\left\{0,1,\perp\right\}$ is a function such that  $(-1)^{\perp}=0$  by convention.

    For the MABK game, the $M$ parties have two measurement settings each, denoted by $x_i\in\{0,1\}$, and all measurements have two possible outcomes, denoted by $a_i\in\{0,1\}$.  
    Then the winning condition for the MABK game is as follows~\cite{ribeiro2017MABK}: 
    \begin{align}\label{eqn:winning_condition_MABK_3}
        \textrm{win}(x_{[M]}, \, a_{[M]})\coloneqq \left\{\begin{array}
                {cc}%
                1,&\text{if}\,\bigoplus_{i=1}^M a_i = f(x_{[M]}) \\%
                0,&\text{else}\,\\
        \end{array}\right.,
    \end{align}
    where $f$ is defined by the $M$-MABK operator in~\eqref{eqn:mabk_operator_rewrite}. The minimum and maximum winning probabilities when the underlying state is genuinely multipartite entangled are respectively as follows~\cite{ribeiro2017MABK}:
    \begin{align}\label{eqn:p_min}
        p_{\min} & \coloneqq  2^{2\lfloor\frac{M}{2}\rfloor -M -1}+2^{\lfloor\frac{M}{2}\rfloor -\frac{M}{2} -2}, \\
        p_{\max} & \coloneqq 2^{2\lfloor\frac{M}{2}\rfloor -M -1}+2^{\lfloor\frac{M}{2}\rfloor -\frac{M}{2} -\frac{3}{2}}.
        \label{eqn:p_max}
    \end{align}
    Note that when $M$ is even, 
    \begin{equation}\label{eqn:p_win_even}
        p^e_{\mathrm{min}}= 3/4 \quad \text{  and  } \quad p^e_{\mathrm{max}}= (2+\sqrt{2})/4.
    \end{equation}
    Similarly when $M$ is odd, 
    \begin{equation}\label{eqn:p_win_odd}
        p^o_{\mathrm{min}}= (2+\sqrt{2})/8 \quad \text{  and  } \quad p^o_{\mathrm{max}}= 1/2.
    \end{equation}
    
    Based on the MABK game defined above, we propose Protocol~\ref{pro:ent_test}. There are two types of rounds in Protocol~\ref{pro:ent_test}: a testing round, where the parties play the MABK game and record the result, and a storage round, where the parties simply store the state they receive in a quantum memory. At the beginning of each round, we choose between these types uniformly random. The main result of our paper is that Protocol~\ref{pro:ent_test} is a DIMEC protocol, satisfying the requirements of  Definition~\ref{def:dimec_protocol}, and we state this result formally in Theorem~\ref{thm:main-result}.
    
    \begin{figure}
    \begin{minipage}{0.95\linewidth}
        \begin{algorithm}[H]
            \caption{: DIMEC protocol based on the MABK game }\label{pro:ent_test}
            \begin{algorithmic}[1]
                \STATEx \textbf{Arguments:}
                \STATEx\hspace{\algorithmicindent} $\mathcal{M}$ -- untrusted measurement device of two components, with inputs and outputs in the set $\{0,1\}$
                \STATEx\hspace{\algorithmicindent} $n \in \mathbb{N}_+$ -- number of rounds
                \STATEx\hspace{\algorithmicindent} $\gamma$ -- the probability of conducting a test
                \STATEx\hspace{\algorithmicindent} $\omega_{\mathrm{exp}}$ -- expected winning probability in the MABK game
                \STATEx\hspace{\algorithmicindent} $\delta_{\mathrm{est}} \in (0,1)$ -- width of the statistical confidence interval for the estimation test
                \STATEx\hspace{\algorithmicindent} $A_{i,j}$ indicates a classical register belonging to the $i$-th party and $j$-th round.
                \STATEx
                \STATE For every round $j\in[n]$, do steps~\ref{prostep:first}-\ref{prostep:last}:
                \STATE\hspace{\algorithmicindent}Let $\phi^j$ denote the multipartite state produced by the source in this round. \label{prostep:first}
                \STATE\hspace{\algorithmicindent}Set $A_{i,j},X_{i,j},W_j=\, \perp$ for all $ i \in [M] $.
                \STATE\hspace{\algorithmicindent}   Choose $T_j=1$ with probability $\gamma$ and $T_j=0$ with probability $1-\gamma$.
                \STATE\hspace{\algorithmicindent}If $T_j=1$:
                \STATE\hspace{\algorithmicindent}   Choose inputs $X_{i,j}\in\{0,1\}$ uniformly at random.
                \STATE\hspace{\algorithmicindent}   Measure $\phi^j$ using $\mathcal{M}$ with the inputs $X_{i,j}$ and record outputs $A_{i,j}\in\{0,1\}$.
                \STATE\hspace{\algorithmicindent}   Set $W_j = 1$ if the MABK game is won and $W_j=0$ otherwise.
                \STATE\hspace{\algorithmicindent}If $T_j=0$:
                \STATE\hspace{\algorithmicindent} Keep $\phi^j$ in the registers $\bigotimes_{i=1}^{m}\hat{A}_{i,j}$. \label{prostep:last}
                \STATE Abort if $ W \coloneqq \sum_{j=1}^n \chi(T_j=1)\,W_j < (\omega_{\mathrm{exp}}\gamma - \delta_{\mathrm{est}})\cdot n$. \label{prostep:abort_ET_1}
            \end{algorithmic}
        \end{algorithm}
    \end{minipage}
    \end{figure}
    
   \section{Completeness}\label{sec:completeness}
   
   In this section, we show that Protocol~\ref{pro:ent_test} is complete. To prove the completeness condition in Definition~\ref{def:dimec_protocol}, we need to show that our protocol applied to an arbitrary $\phi^{\operatorname{honest}}\in\mathcal{S}^{\operatorname{honest}}$ aborts with only a small probability. We prove this result in the following theorem.
    \begin{theorem}\label{thm:completeness}
    The following bound holds: 
        \begin{equation}\label{eqn:completeness_thm}
            \varepsilon_{\operatorname{cmp}}\leq \exp\!\left(-2n\delta^{2}_{\operatorname{est}}\right),
        \end{equation}
        implying that the probability that Protocol~\ref{pro:ent_test} aborts for $\phi^{\operatorname{honest}}\in\mathcal{S}^{\operatorname{honest}}$ is no larger than $\exp\!\left(-2n\delta^{2}_{\operatorname{est}}\right)$.
    \end{theorem}
    
    \begin{proof}
        Protocol~\ref{pro:ent_test} aborts  when $W$, the estimator for the MABK game winning probability, is not high enough. This happens when the MABK violation is not large or  if the number of samples is not sufficiently large. Whenever  $\phi^{\operatorname{honest}}\in\mathcal{S}^{\operatorname{honest}}$ and the rounds are independent and identically distributed (IID), i.e., an honest implementation, the sequence $\left(\chi(T_j=1)\,W_j\right)_{j=1}$ is a sequence of IID random variables. Using Hoeffding's inequality \cite{H63}, we conclude that
        \begin{equation}\label{eqn:completeness_epsilon}
            \varepsilon_{\operatorname{cmp}}=\Pr\!\left(W \leq (\omega_{\mathrm{exp}}\gamma - \delta_{\mathrm{est}})\cdot n\right)\leq e^{-2n\delta^{2}_{\operatorname{est}}},
        \end{equation}
        completing the proof.
    \end{proof}

    

    \section{Soundness}\label{sec:soundness}

    In this section, we prove the soundness of Protocol~\ref{pro:ent_test}. To do so, we need to show that the protocol ensures a minimum amount of multipartite distillable entanglement or aborts with high probability. First, we briefly discuss a protocol to distill multipartite entanglement from a mixed state. We will use this achievable rate of multipartite distillable entanglement to prove the soundness of Protocol~\ref{pro:ent_test}.

    Consider the following setting: a pure state $\psi_{A_1\cdots A_M B R}$ is held by $M+2$ parties, with the $i$th  sender holding $A_i$, for $i\in\{1,\ldots,M\} \eqqcolon [M]$, one receiver holding $B$, and a reference system holding $R$, which serves as a purifying system. Additionally, Alice-$i$ and Bob share a maximally entangled state $\Phi(c_i)_{A^\prime_i B^\prime_i}$ of Schmidt rank $c_i \in \mathbb{N}$, and so the initial overall state is  
    \begin{equation}\label{eqn:init_state_state_merge}
        \psi_{A_1\cdots A_M B R}\otimes \Phi(c_1)_{A^\prime_1 B^\prime_1} \otimes \cdots \otimes \Phi(c_M)_{A^\prime_M B^\prime_M}.
    \end{equation}
    Multiparty state merging is an LOCC protocol  that transforms this initial state to 
    \begin{equation}\label{eqn:final_state_state_merge}
        \psi_{\hat{A}_1\cdots \hat{A}_M \hat{B} R}\otimes \Phi(d_1)_{A^{\prime\prime}_1 B^{\prime\prime}_1} \otimes \cdots \otimes \Phi(d_M)_{A^{\prime\prime}_M B^{\prime\prime}_M},
    \end{equation}
    such all of the systems $\hat{A}_1\cdots \hat{A}_M \hat{B}$ are in Bob's possession.
    (For details, see Section~6.3 of~\cite{colomer2023decoupling}.)
    If $c_i>d_i$, then the protocol consumes $\log_2 c_i-\log_2 d_i$ ebits shared between $A_i$ and $B$. If $c_i<d_i$, then the protocol produces $\log_2 d_i-\log_2 c_i$ ebits, which are shared between $A_i$ and $B$ (recall that the term ``ebit'' is synonymous with the standard Bell pair). If the protocol distills ebits between each Alice, $A_1\cdots A_M$, and Bob~$B$, then they can apply LOCC to convert several copies of their ebits into multipartite GHZ states.
    
    For a one-shot realization of the entire process of using LOCC to convert $\psi_{A_1\cdots A_M B R}$ into a GHZ state, we need only understand the one-shot behavior of the state merging step. The theorem needed for the state merging step is as follows:

    \begin{theorem}[Theorem~6.15 of~\cite{colomer2023decoupling}]\label{thm:state_merge_one_shot}
        Given the setting above, quantum state merging can be achieved successfully, in the sense that the fidelity constraint in~\eqref{eqn:one_shot_distillable_ent} holds, for 
        \begin{equation}
        \varepsilon =8\cdot3^{M/2}\sqrt{\varepsilon^\prime}  \in (0,1),
        \end{equation}  if for all $K\subseteq [M]$ such that $K\neq\emptyset$,
        \begin{multline}\label{eqn:one_shot_state_merge}
             \sum_{k\in K}(\log_2 d_k-\log_2 c_k)\\ \leq -H^{\varepsilon^\prime}_{\max}(A_K|A_{K^c}B)_{\psi}+2\log\varepsilon^\prime.
        \end{multline}
    \end{theorem} 
    
    Using the above theorem and the steps elucidated earlier, we conclude the following:
    \begin{theorem}[Theorem~9 of \cite{salek2023new}]\label{theorem:asymptot}
        Let $\rho_{A_1\cdots A_M}$ be a state held by $m$~parties. The following is an achievable rate for GHZ distillation under LOCC:
        \begin{equation}\label{eqn:asymptotic_lower_bound}
            \max_{k\in [M]}\left\{\min_{K\subseteq [M]\backslash k}\frac{I(A_K\rangle A_{[M]\backslash K})_\rho}{|K|}\right\},
        \end{equation}
        where $I(A\rangle B)_\theta$ is the coherent information of a state~$\theta_{AB}$. 
    \end{theorem}
    
    The statement above gives a bound only on the asymptotic rate of GHZ distillation from mixed states. For our purpose, we require a lower bound on $\frac{1}{n}E_{D}^{\varepsilon}(\rho_{A_{[M]}}^{\otimes n})$. The following proposition gives such a bound by making use of the developments of \cite{colomer2023decoupling,salek2023new}.
    \begin{proposition}
        Let $\rho_{A_1\cdots A_M}$ be held by $m$ parties. Fix $\varepsilon' > 0$ such that $\varepsilon \coloneqq 8\cdot 3^{M/2} \sqrt{\varepsilon^{\prime}}  \in (0,1)$. Then the following is an achievable one-shot rate for GHZ distillation under LOCC; i.e.,
        \begin{multline}\label{eqn:one_shot_lower_bound}
        \frac{1}{n}E_{D}^{\varepsilon}(\rho_{A_{[M]}}^{\otimes n}) \geq \\
            \max_{k\in [M]}\left\{\min_{K\subseteq [M]\backslash k}\frac{-H^{\varepsilon^\prime}_{\max}(A_K|A_{[M]\backslash K})_{\rho}}{|K|}\right\}+\frac{2\log\varepsilon^\prime}{M-1}.
        \end{multline}
    \end{proposition}
    
    \begin{proof}
        To achieve this one-shot rate, we can use the protocol outlined in \cite[Theorem~9]{salek2023new} but without time sharing. Time shared decoding is replaced with simultaneous decoding, as the former is not possible in the one-shot setting. We obtain achievable rates for the one-shot case by replacing coherent information in~\eqref{eqn:asymptotic_lower_bound} with the conditional smooth max-entropy. This leaves us with an expression commensurate with~\eqref{eqn:one_shot_state_merge} in Theorem~\ref{thm:state_merge_one_shot}. 
    \end{proof}

\subsection{Entropy Accumulation Theorem (EAT)}

    To certify that a minimum amount of multipartite entanglement can be distilled from the states left over at the end of $n$ rounds of Protocol~\ref{pro:ent_test}, we need to lower bound  $-H_{\max}^{\varepsilon}\!\left(A^{\otimes n}_K|A^{\otimes n}_{[M]\backslash K}\right)_{\rho^{\otimes n}}$. Since this is a finite-length protocol and we are interested in a smooth entropic quantity, we can use Entropy Accumulation Theory (EAT) to obtain a rate~$r$ in terms of the desired error constants. To apply EAT, one needs to define ``EAT channels,'' which describe the sequential process under consideration, and max-tradeoff functions. In our case, the sequential process results from our protocol for the certification of multipartite entanglement distillation.
    
    For a sequential process like our protocol, the most restrictive relation between the rounds is to assume that they are IID. However, such a restriction is arguably too strong for the device-independent scenario. Entropy accumulation theory allows for a more general relation between the rounds of our protocol. This is encapsulated in the following definition of EAT channels.

    \begin{definition}[EAT channels \cite{dupuis2016entropy}]\label{def:eat_channels} 
        An EAT channel $\mathcal{N}_{j}:R_{j-1}\rightarrow R_{j} O_{j} S_{j} W_{j}$, for $j\in[n]$, is a CPTP map, such that, for all $j\in[n]$:

        \begin{enumerate}
        \item \label{cond:eat_chan_1} $W_{j}$ is a finite-dimensional classical system. $S_{j}$ and $R_{j}$ are 
         arbitrary quantum systems.

        \item \label{cond:eat_chan_2} Given an input state $\sigma_{R_{j-1}}$, the output state $\sigma_{R_{j} O_{j} S_{j}W_{j}} = \mathcal{N}_{j} \!\left(\sigma_{R_{j-1}} \right)$ has the property that one can perform a quantum instrument on the systems $O_j S_j$ (in the state $\sigma_{O_{j}S_{j}}$), obtain the classical register $W_{j}$, and discard it,  without changing the state $\sigma_{O_{j}S_{j}}$. That is, for the instrument $\mathcal{T}_{j} :O_{j}S_{j}\rightarrow O_{j}S_{j}W_{j}$ describing the process of obtaining $W_{j}$ from $O_{j}$ and $S_{j}$, it holds that $(\operatorname{Tr}_{W_{j}}\circ\mathcal{T}_{j}) \left(\sigma_{O_{j}S_{j}}\right)  = \sigma_{O_{j}S_{j}} $.

        \item \label{cond:eat_chan_3} $O_{j}$ is a finite-dimensional quantum system of dimension~$d_{O_{j}}$.

        \end{enumerate}
    \end{definition}

    If the rounds are not IID, then we cannot consider a round of the protocol to be completely independent of the preceding rounds. In other words, we cannot use the additivity property of the quantities under consideration. To resolve this, entropy accumulation theory requires that a special function, associated with our protocol, be found. This special function is called the max-tradeoff function.  

    To define max-tradeoff functions, we need to clarify some notation, which we will use henceforth. Given a value $w=(w_{1}, \dotsc, w_{n}) \in\mathcal{W}^{n}$, where $\mathcal{W}$ is a finite alphabet, we denote by $\mathrm{freq}_{w}$ the probability distribution over $\mathcal{W}$ defined by $\mathrm{freq}_{w}(\tilde{w}) \coloneqq \frac{\left| \left\{  j | w_{j} = \tilde{w}\right\}  \right|}{n}$ for $\tilde{w}\in\mathcal{W}$. If $\tau$ is a state classical on $W$, we write $\Pr\left[  w\right]  _{\tau}$ to denote the probability that $\tau$ assigns to $w$. Now, we move on to the definition of max-tradeoff functions.

    \begin{definition}[Max-tradeoff functions \cite{dupuis2016entropy}] \label{def:max_tradeoff_func}
        Let $\mathcal{N}_{1}$, \ldots, $\mathcal{N}_{n}$ be a family of EAT channels. Let $\mathcal{W}$ denote the common alphabet of $W_{1},\ldots,W_{n}$. A concave \footnote{Let $\hat{\Omega}$ be a set of frequencies defined via $\mathrm{freq}_{w} (\tilde{w})\in\hat{\Omega}$ if and only if $\tilde{w}\in\Omega$. We can consider concave functions, in contrast to affine ones~\cite{dupuis2016entropy}, since the event $\Omega$ defined in the current work results in a \emph{convex set} $\hat{\Omega}$.} function $f_{\max}$  from the set of probability distributions~$p$ over $\mathcal{W}$ to the real numbers is called a max-tradeoff function for $\{\mathcal{N}_{j}\}_j$ if it satisfies
        \begin{equation}\label{eqn:max_tradeoff_func_def}
            f_{\max}(p) \geq\sup_{\sigma} \left\{H\!\left(  O_{j} \middle\vert S_{j} \right)  _{\mathcal{N}_{j}(\sigma)}:\operatorname{Tr}_{R_j O_j S_j}[\mathcal{N}_{j}(\sigma)]=p\right\},
        \end{equation}
        for all $j\in[n]$, where the supremum is taken over all input states of $\mathcal{N}_{j}$ for which the marginal on $W_{j}$ of the output state is the probability distribution $p$.
    \end{definition}

    The statement of the EAT, relevant for the smooth max-entropy, is given below (see \cite[Theorem~4.4]{dupuis2016entropy} and Eq.~(A.2) of \cite[Appendix~A]{9996821}). 

    \begin{theorem}[\cite{dupuis2016entropy,9996821}]\label{thm:eat} 
        Let $\mathcal{N}_{j}:R_{j-1}\rightarrow R_{j} O_{j} S_{j} W_{j}$ for $j\in[n]$ be a sequence of EAT channels as in Definition~\ref{def:eat_channels}, $\tau_{OSW} = \left(  \operatorname{Tr}_{R_{n}}\circ\mathcal{N}_{n} \circ \dots\circ\mathcal{N}_{1}\right)  (\tau_{R_{0}})$ the final state, $\Omega$ an event defined over $\mathcal{W}^{n}$ indicating acceptance, $\Pr[\Omega]_{\tau}$ the probability of acceptance  given the underlying state $\tau$, and $\tau_{|\Omega}$ the final state conditioned on $\Omega$. 

        Let $\varepsilon_{\mathrm{smo}} \in(0,1)$. For $f_{\max}$ a max-tradeoff function for $\{\mathcal{N}_{j}\}_j$, as in Definition~\ref{def:max_tradeoff_func}, and all $t\in\mathbb{R}$ such that $f_{\max}\left(  \mathrm{freq}_{w} \right)  \leq t$ for all $w\in \mathcal{W}^{n}$ for which $\Pr\left[  w\right]  _{\tau_{|\Omega}}> 0$, the following holds:
        \begin{equation}\label{eqn:h_max_eat}
            H_{\max}^{\varepsilon_{\mathrm{smo}}} \left(  O|S \right)  _{\tau_{|\Omega}}\leq n t + v\sqrt{n} \;,
        \end{equation}
        where $d_{O_{i}}$ denotes the dimension of~$O_{i}$ and 
        \begin{multline}\label{eqn:h_max_eat_1}
            v  \coloneqq 2\left(  \log_2(1+2 d_{O_{i}} ) 
            + \lceil\left\| \triangledown f_{\max}\right\|_{\infty}\rceil\right)\\ \qquad\times  \sqrt{1-2\log_2(\varepsilon_{\mathrm{smo}} \cdot \Pr[\Omega]_{\tau})}.
        \end{multline}
    \end{theorem}
    
    It has been shown that Theorem~\ref{thm:eat}, the generalized entropy accumulation theorem, holds regardless of whether the sequence of channels $\{\mathcal{N}_{j}\}_j$ satisfies a Markov-chain condition \cite[Appendix~A]{9996821}. Note that the original entropy accumulation theorem~\cite{dupuis2016entropy} does require such a Markov-chain condition.

\subsection{EAT Channels}\label{subsec:eat_channels}

    In this subsection, we prove that Protocol~\ref{pro:ent_test}, for all  $j\in[n]$, satisfies the conditions for EAT channels, $\mathcal{N}_{j}:R_{j-1}\rightarrow R_{j} O_{j} S_{j} W_{j}$ as outlined in Definition~\ref{def:eat_channels}. This is a necessary condition for the application of the entropy accumulation theorem~\cite{dupuis2016entropy}. 
    
    Consider Condition~\ref{cond:eat_chan_1} of Definition~\ref{def:eat_channels}. We can think of $R_j$ as a source distributing the states across the network. As far we are concerned, it is an arbitrary quantum system. $W_j$ is the classical register associated with determining whether the MABK game has been won. $W_j$ is finite dimensional as it can take only three values. $S_j$ consists of all the systems involved in the protocol, except for the system $\hat{A}_{[M^\prime],\,j}$ for some $M^\prime \in [M-1]$.
    
    For Condition~\ref{cond:eat_chan_2} of  Definition~\ref{def:eat_channels}, we can see that $W_j$ is determined using only the classical values $A_{i,j}$ and $X_{i,j}$ and does not affect the classical and quantum registers.
    
    Finally, we are left with Condition~\ref{cond:eat_chan_3} of Definition~\ref{def:eat_channels}. Checking these conditions is more involved. First, we note that Protocol~\ref{pro:ent_test} makes use of the MABK inequality, which involves two inputs and two outputs for all parties involved. This allows us to apply Jordan's lemma on every party. Now, we recall Jordan's lemma~\cite{scarani20124}:
    
    \begin{theorem}[Lemma 4.1 in~\cite{scarani20124}]\label{thm:jordan}
        Let $\hat{O}_0$ and $\hat{O}_1$ be two Hermitian operators with eigenvalues $-1$ and $+1$. Then there exists a basis in which both operators are block diagonal, in blocks of dimension $2\times 2$ at most.
    \end{theorem}
    
     For each party $A_i$ and for every round $j$, we can reduce the associated binary observables $\hat{O}^{i,j}_0$ and $\hat{O}^{i,j}_1$ to a block-diagonal form in a suitable local basis using Theorem~\ref{thm:jordan}. The block-diagonal form is as follows:
    \begin{align}\label{eqn:jordan_blocks}
        \hat{O}^{i,j}_0 &= \bigoplus_{d_{i,j}}\, O_{0}^{d_{i,j}}= \bigoplus_{d_{i,j}}\,\sigma_{y}^{d_{i,j}},\\
        \hat{O}^{i,j}_1 &= \bigoplus_{d_{i,j}}\, O_{1}^{d_{i,j}} \\ 
        & =\bigoplus_{d_{i,j}}\left( \cos(\alpha_{d_{i,j}})\,\sigma_{y}^{d_{i,j}}+\sin(\alpha_{d_{i,j}})\,\sigma_{x}^{d_{i,j}}\right).
    \end{align}

Let $\Pi_{d_{i,j}} $ denote the projection onto the $d_{i,j}$th block.
    If we act with the projection $\Pi_{d_{i,j}} $ on $\rho_{\hat{A}_{[M],j}}$, then we obtain the index $d_{i,j}$ of the corresponding Jordan block. Hence, after the application of the projection, the system possessed by each party is a two-dimensional system. This is true of $\hat{A}_{i,j}$ especially, which satisfies Condition~\ref{cond:eat_chan_3}. After party $i$ performs the measurement $\{\Pi_{d_{i,j}}\}_{d_{i,j}}$ for round $j$, for all $i\in [M]$, the post-measurement state is as follows:
        \begin{equation}\label{eqn:projected_state}    \frac{\left(\bigotimes_{i=1}^{M}\Pi_{d_{i,j}}\right)\rho_{\hat{A}_{[M],j}} \left(\bigotimes_{i=1}^{M}\Pi_{d_{i,j}}\right)}{\operatorname{Tr}\left[\left(\bigotimes_{i=1}^{M}\Pi_{d_{i,j}}\right)\rho_{\hat{A}_{[M],j}}\right]}
        \otimes\left(\bigotimes_{i=1}^{M}|d_{i,j}\rangle\! \langle d_{i,j}|\right).
        \end{equation}
    Note that all quantum registers have been reduced to qubit registers.
    
    Using these projections, we produce Protocol~\ref{pro:ent_test_2}. Protocol~\ref{pro:ent_test_2} is a device-dependent version of Protocol~\ref{pro:ent_test} but has the same winning probability as Protocol~\ref{pro:ent_test}. This is due the fact that, regardless of the value of $T_j$ and for any $d_j=d_{1,j}\cdots d_{M,j}$, the underlying state can thought of as being in the following state~\cite[Theorem~1]{PRXQuantum.2.010308}:
    \begin{multline}\label{eqn:ghz_incoherent_state}
        \tilde{\rho}^{d_j} \coloneqq \sum_{u\in\{0,1\}^{M-1}} \Big[\lambda^{d_j}_{0,u}|\psi_{0,u}\rangle\!\langle \psi_{0,u}|  + \lambda^{d_j}_{1,u}|\psi_{1,u}\rangle\!\langle \psi_{1,u}| \\
         +\mathbf{i}s^{d_j}_{u}\left(|\psi_{0,u}\rangle\!\langle \psi_{1,u}|-|\psi_{1,u}\rangle\!\langle \psi_{0,u}|\right)\Big], 
    \end{multline}
    where $ \mathbf{i}^2=-1$, $s^{d_j}_{u}, \lambda^{d_j}_{0,u}\in [0,1]$, and $\psi_{0,u}, \psi_{1,u}$ are defined in Definition~\ref{def:ghz_basis}. Hence, the projective measurement detailed earlier will neither change the winning probability nor the amount of multipartite entanglement produced by the source. (This is true for any two-input two-output, or ($M$, 2, 2), full-correlator Bell inequality~\cite[Theorem~1]{PRXQuantum.2.010308}, like the MABK inequality we consider here.) 
    
    We will use Protocol~\ref{pro:ent_test_2} only to make statements about the soundness of Protocol~\ref{pro:ent_test}. 
    First, we need to clarify some notation. As before, the index~$i$ will denote the party, and index~$j$ will denote the round of the protocol. $\hat{A}_{[M]\,[n]}$ refers to all the quantum registers $\hat{A}_{i,j}$ possessed by the respective parties at the end of the protocol before conditioning on the outcome of the protocol. $A_{[M]\,[n]}$ refers to the classical registers $A_{i,j}$ containing the inputs of the protocol used by the parties in each round, regardless of the $T_j$ of said rounds. $D_{[M]\,[n]}$ contains the results of the Jordan block projection $D_{i,j}$ performed in each round of the protocol. $X_{[M]\,[n]}$ refers to the classical registers containing output $X_{i,j}$ from all the rounds of the protocol. For convenience, we shall refer to $A_{[M]\,[n]}$, $D_{[M]\,[n]}$, and $X_{[M]\,[n]}$ together as $(ADX)_{[M]\,[n]}$. $\overline{W}$ and $T$ refer to the classical registers containing $W_{j}$ and $T_{j}$, respectively, from all the rounds of the protocol. 

    \begin{figure}
    \begin{minipage}{0.95\linewidth}
        \begin{algorithm}[H]
            \caption{: DIMEC protocol based on the MABK game}\label{pro:ent_test_2}
            \begin{algorithmic}[1]
                \STATEx \textbf{Arguments:}
                \STATEx\hspace{\algorithmicindent} $\mathcal{M}$ -- untrusted measurement device of two components, with inputs and outputs in the set $\{0,1\}$
                \STATEx\hspace{\algorithmicindent} $n \in \mathbb{N}_+$ -- number of rounds
                \STATEx\hspace{\algorithmicindent} $\gamma$ -- the probability of conducting a test
                \STATEx\hspace{\algorithmicindent} $\omega_{\mathrm{exp}}$ -- expected winning probability in the MABK game
                \STATEx\hspace{\algorithmicindent} $\delta_{\mathrm{est}} \in (0,1)$ -- width of the statistical confidence interval for the estimation test
                \STATEx\hspace{\algorithmicindent} $A_{i,j}$ indicates a classical register belonging to the $i$-th party and $j$-th round.
                \STATEx
                \STATE For every round $j\in[n]$, do steps~\ref{prostep:first_2}-\ref{prostep:last_2}:
                \STATE\hspace{\algorithmicindent}Let $\phi^j$ denote the multipartite state produced by the source in this round. \label{prostep:first_2}
                \STATE\hspace{\algorithmicindent}Set $A_{i,j},X_{i,j},W_j=\, \perp$ for all $ i \in [M] $.
                \STATE\hspace{\algorithmicindent}Choose $T_j=1$ with probability $\gamma$ and $T_j=0$ with probability $1-\gamma$.
                \STATE\hspace{\algorithmicindent}If $T_j=1$:
                \STATE\hspace{\algorithmicindent} Choose inputs $X_{i,j}\in\{0,1\}$ uniformly at random.
                \STATE\hspace{\algorithmicindent} Measure $\phi^j$ using $\mathcal{M}$ with the inputs $X_{i,j}$ and record outputs $A_{i,j}\in\{0,1\}$.
                \STATE\hspace{\algorithmicindent} Set $W_j = 1$ if the MABK game is won and $W_j=0$ otherwise.
                \STATE\hspace{\algorithmicindent}If $T_j=0$:
                \STATE\hspace{\algorithmicindent} Apply projections described in~\eqref{eqn:projected_state}.
                \STATE\hspace{\algorithmicindent} Keep $\phi^j$ in the registers $\bigotimes_{i=1}^{m}\hat{A}_{i,j}$. \label{prostep:last_2}
                \STATE Abort if $W =\sum_{j=1}^n \chi(T_j=1)\,W_j < (\omega_{\mathrm{exp}}\gamma - \delta_{\mathrm{est}}) n$. \label{prostep:abort_ET_2}
            \end{algorithmic}
        \end{algorithm}
        \end{minipage}
    \end{figure}

\subsection{Max-Tradeoff Function}\label{subsec:max_tradeoff}

    In this section, we obtain a max-tradeoff function that satisfies Definition~\ref{def:max_tradeoff_func}. We are interested in finding an upper bound on the following quantity:
    \begin{equation}
        \sup_{\sigma}H\!\left(\hat{A}_{[M^\prime],j} \middle\vert \hat{A}_{[M^\prime+1,M],j}\,(ADX)_{[M],j}\,T_j\right)_{\mathcal{N}_{j}(\sigma)},
        \label{eq:cond-ent-proof}
    \end{equation}
    for all $j\in[n]$ and $M^\prime\in[M-1]$, where the supremum is taken over all input states of $\mathcal{N}_{j}$ for which the marginal on $W_{j}$ of the output state is the probability distribution~$p$. Henceforth, we shall refer this set of states as 
    \begin{equation}\label{eqn:sigma_p}
        \Sigma_p \coloneqq\left\{\sigma\,\middle |\,\mathcal{N}_{j}(\sigma)_{W_j}=p\right\}.
    \end{equation}

    We will first simplify~\eqref{eq:cond-ent-proof}. Note that the following equality holds:
    \begin{equation}
    \label{eq:cond-ent-cq}
        \sum_x p(x) H(A|B)_{\rho^x} =  H(A|BX)_{\tilde{\rho}},
    \end{equation}
    where $\{p(x)\}_x$ is a probability distribution, $\{\rho^x_{AB}\}_x$ is a set of states, $\tilde{\rho}_{ABX} \coloneqq \sum_x p(x) \rho_{AB}^x \otimes |x\rangle\!\langle x|$, and $\{|x\rangle\}_{x}$ is an orthonormal basis. Using this, we find that
    \begin{align}
        &\sup_{\sigma\in\Sigma_p}H\!\left(\hat{A}_{[M^\prime],j} \middle\vert \hat{A}_{[M^\prime+1,M],j}\,(ADX)_{[M],j}\,T_j\right)_{\mathcal{N}_{j}(\sigma)} = \notag\\
        &\sup_{\sigma\in\Sigma_p} \Bigg[\mathrm{Pr}(T_j=0)H\!\left(\hat{A}_{[M^\prime],j} \middle\vert \hat{A}_{[M^\prime+1,M],j}\,(ADX)_{[M],j}\right)_{\mathcal{N}_{j}^{0}(\sigma)}\notag\\
        &+\mathrm{Pr}(T_j=1)H\!\left(\hat{A}_{[M^\prime],j} \middle\vert \hat{A}_{[M^\prime+1,M],j}\,(ADX)_{[M],j}\right)_{\mathcal{N}_{j}^{1}(\sigma)}\Bigg],
    \end{align}
    where $\mathcal{N}_{j}^{0}(\sigma)$ corresponds to the state when $T_j=0$, and $\mathcal{N}_{j}^{1}(\sigma)$ corresponds to the state when $T_j=1$. Note that $T_j=1$ corresponds to a round where we apply measurements to play the MABK game and each $\hat{A}_{i,j}$ is in a deterministic state for all $i\in[M]$. This implies that 
    \begin{equation}
    H\!\left(\hat{A}_{[M^\prime],j} \middle\vert \hat{A}_{[M^\prime+1,M],j}\,(ADX)_{[M],j}\right)_{\mathcal{N}_{j}^{1}(\sigma)} = 0.
    \end{equation}
    Hence,
    \begin{multline}
        \sup_{\sigma\in\Sigma_p}H\!\left(\hat{A}_{[M^\prime],j} \middle\vert \hat{A}_{[M^\prime+1,M],j}\,(ADX)_{[M],j}\,T_j\right)_{\mathcal{N}_{j}(\sigma)}= \\
        (1-\gamma)\sup_{\sigma\in\Sigma_p}H\!\left(\hat{A}_{[M^\prime],j} \middle\vert \hat{A}_{[M^\prime+1,M],j}\,(ADX)_{[M],j}\right)_{\mathcal{N}_{j}^{0}(\sigma)}.
    \end{multline}
    
    To calculate an upper bound on the above quantity, we will use the fact that the conditional entropy is concave \cite{10.1063/1.1666274}:
    \begin{equation}
        \sum_x p(x) H(A|B)_{\rho^x}\leq H(A|B)_{\bar{\rho}},
    \end{equation}
    where the notation is the same as in \eqref{eq:cond-ent-cq} and $\bar{\rho}_{AB} \coloneqq \sum_x p(x) \rho_{AB}^x$. 
    We will also make use of the bipartition between systems represented by the vertical bar in~\eqref{eq:cond-ent-proof}, which means that we can make use of local operations with respect to this bipartition. Given a particular $D_{[M],j}$ at the end of the $j$th round of Protocol~\ref{pro:ent_test_2}, the state $\tilde{\rho}^{d_j}$ in~\eqref{eqn:ghz_incoherent_state} has a simple structure in the GHZ basis. We can rewrite the GHZ basis, from Definition~\ref{def:ghz_basis}, in the following fashion: 
    \begin{equation}\label{eqn:def_ghz_basis_alt}
        |\psi_{v,u,w}\rangle \equiv \frac{1}{\sqrt{2}}\left[|w,u\rangle +(-1)^v |\bar{w},\bar{u}\rangle\right],
    \end{equation}
    where $v\in \{0,1\}$, while $u\in\{0,1\}^{M-M^\prime}$ and $\bar{u}=1\oplus u$ are $M-M^\prime$ bit strings, and $w\in\{0,1\}^{M^\prime}$ and $\bar{w}=1\oplus w$ are $M^\prime$ bit strings. 
    
    Within the partition denoted by $u$, we can choose one qubit to be the control qubit, $u_1$, and then apply CNOTs on all the other qubits. We apply similar set of operations on the partition denoted by $w$ with the first qubit containing $w_1$ as the control qubit. These operations transform the $M$-partite GHZ basis states into the following state:
    \begin{multline}
        \frac{1}{2}\left(|u_1,w_1\rangle +(-1)^v |\bar{u}_1,\bar{w}_1\rangle\right)\otimes|u_2\cdots u_{M^\prime}\rangle\\ \otimes|w_2\cdots w_{M-M^\prime}\rangle, 
    \end{multline}
    where $u_i, w_i\in\{0,1\}$. The above state is a Bell-basis state between the control qubits containing $u_1$ and $w_1$. After the CNOTs, as described above, are applied on the state $\tilde{\rho}^{d_j}$ in~\eqref{eqn:ghz_incoherent_state}, there are still off-diagonal terms in the resulting state. To make the final state diagonal in the Bell basis, we use the following twirling channel \cite{PhysRevA.54.3824}:
    \begin{multline}
        \mathcal{T}(\rho)=\frac{1}{4}\Big(\rho + (X\otimes X)\rho\,(X\otimes X) +\\ (Y\otimes Y)\rho\,(Y\otimes Y) +  (Z\otimes Z)\rho\,(Z\otimes Z)\Big).
    \end{multline}
    Now that we have simplified the underlying state, we make the following observation about the MABK inequality.
    \begin{lemma}\label{thm:mabk_2_bell}
    An $M$-partite MABK inequality~\eqref{eqn:inequality_mabk} is equivalent to the CHSH inequality within any bipartition consisting of $M^\prime$ parties on  one side and the $M -M^\prime$ other parties on the other side such that, for some $\hat{O}^{[M-M^\prime+1]}_0$ and $\hat{O}^{[M-M^\prime+1]}_1$, $\mathcal{K}_M$ is equal to
    \begin{multline}
        \frac{1}{2}\mathcal{F}\!\left(\mathcal{K}_{M-M^\prime},\overline{\mathcal{K}}_{M-M^\prime},\hat{O}^{[M-M^\prime+1]}_0, \hat{O}^{[M-M^\prime+1]}_1\right).
    \end{multline}
    \end{lemma}
\begin{proof}
    See Appendix~\ref{proof:mabk_2_bell}. 
\end{proof}
\medskip
    
    Given the above operations and Lemma~\ref{thm:mabk_2_bell}, we can use a prior result from \cite[Lemma 14]{Arnon_Friedman_2019} to obtain our max-tradeoff function. We can replace the CHSH violation $\beta_2$ in~\eqref{eqn:inequality_chsh_limits} with the multipartite MABK violation $\beta_M$ in~\eqref{eqn:inequality_mabk_limits} for our case. We restate \cite[Lemma 14]{Arnon_Friedman_2019} for convenience.
    
    \begin{lemma}[Lemma 14 of \cite{Arnon_Friedman_2019}]\label{thm:bell_diag}
        For every Bell-diagonal state $\sigma_{\hat{A}\hat{B}}$ that can be used to violate the CHSH inequality with violation $\beta_2\in[2,2\sqrt{2}]$, the following inequality holds:
        \begin{equation}
            H(\hat{A}|\hat{B}) \leq 2h_2\!\left(\frac{1}{2}-\frac{\beta_2}{4\sqrt{2}}\right)-1,
        \end{equation}
        where $h_2(p) \coloneqq -p \log_2 p - (1-p) \log_2(1-p)$ is the binary entropy function and $\beta_2$ is defined in~\eqref{eqn:inequality_chsh_limits}.
    \end{lemma} 

    These lemmas can be applied to Protocol~\ref{pro:ent_test_2} in the following manner. 
    Using Lemma~\ref{thm:mabk_2_bell} and Lemma~\ref{thm:bell_diag}, we find that
    \begin{multline}\label{eqn:max_tradeoff_interm}
        \sup_{\sigma\in\Sigma_p}H\!\left(\hat{A}_{[M^\prime],j} \middle\vert \hat{A}_{[M^\prime+1,M],j}\,(ADX)_{[M],j}\,T_j\right)_{\mathcal{N}_{j}(\sigma)} \\
        \leq(1-\gamma)\left(2h_2\!\left(\frac{1}{2}-\frac{\beta_M}{4\sqrt{2}}\right)-1\right),
    \end{multline}
    where $\beta_M$ is defined in~\eqref{eqn:inequality_mabk_limits}. Using \eqref{eqn:max_tradeoff_interm}, we can state the following theorems.
    \begin{theorem}\label{thm:max_func_even}
        For all $M^\prime\in[M]$ with $M$ even and $\omega\in [p^e_{\min},p^e_{\max}]$ as defined in~\eqref{eqn:p_win_even}, respectively, let $\Sigma_p$ (defined in~\eqref{eqn:sigma_p}) be such that the winning probability is strictly greater than $\omega$ (i.e., $p(1)/\gamma>\omega$). Then,
        \begin{multline}\label{eqn:max_func_almost}
            \sup_{\sigma\in\Sigma_p}H\!\left(\hat{A}_{[M^\prime],j} \middle\vert \hat{A}_{[M^\prime+1,M],j}\,(ADX)_{[M],j}\,T_j\right)_{\mathcal{N}_{j}(\sigma)} \\ \leq (1-\gamma)\cdot g_e(\omega),
        \end{multline}
        where $\gamma $ is defined in Protocol~\ref{pro:ent_test}, 
        \begin{equation}\label{eqn:bin_ent_even}
            g_e(\omega) \coloneqq 2h_2\!\left(\frac{1}{2}-\frac{2\omega-1}{\sqrt{2}}\right)- 1,
        \end{equation}
        and $h_2(p)$ is the binary entropy function.
    \end{theorem}
    \begin{proof}
        When $M$ is even, the MABK game winning probability and the MABK violation in~\eqref{eqn:inequality_mabk} are related as follows~\cite{ribeiro2017MABK}: $\beta_M= 8\omega -4$. Substituting $\beta_M= 8\omega -4$ in~\eqref{eqn:max_tradeoff_interm}, we get the required result.
    \end{proof}
    \medskip
    
    \begin{theorem}\label{thm:max_func_odd}
        For all $M^\prime\in[M]$ with $M$ odd and $\omega\in [p^o_{\min},p^o_{\max}]$ as defined in~\eqref{eqn:p_win_odd}, respectively, let $\Sigma_p$ (defined in~\eqref{eqn:sigma_p}) be such that the winning probability is strictly greater than $\omega$ (i.e., $p(1)/\gamma>\omega$). Then,
        \begin{multline}\label{eqn:max_func_odd}
            \sup_{\sigma\in\Sigma_p}H\!\left(\hat{A}_{[M^\prime],j} \middle\vert \hat{A}_{[M^\prime+1,M],j}\,(ADX)_{[M],j}\,T_j\right)_{\mathcal{N}_{j}(\sigma)} \\ \leq (1-\gamma)\cdot g_o(\omega),
        \end{multline}
        where $\gamma $ is defined in Protocol~\ref{pro:ent_test},
        \begin{equation}\label{eqn:bin_ent_odd}
            g_o(\omega) \coloneqq 2 h_2\!\left(\frac{1}{2}-\frac{4\omega-1}{2}\right)- 1,
        \end{equation}
        and $h_2(p)$ is the binary entropy function.
    \end{theorem}
    \begin{proof}
        When $M$ is odd, the MABK game winning probability and the MABK violation in~\eqref{eqn:inequality_mabk} are related as follows~\cite{ribeiro2017MABK}: $\beta_M=8\sqrt{2}\omega-2\sqrt{2}$. Substituting $\beta_M=8\sqrt{2}\omega-2\sqrt{2}$ in~\eqref{eqn:max_tradeoff_interm}, we get the required result.
    \end{proof}
    \medskip
    
    These theorems serve as the basis for our max-tradeoff function. Using Theorems~\ref{thm:max_func_even} and~\ref{thm:max_func_odd} to define a max-tradeoff function for all $p$ with $p(1)/\gamma\in [\,p_{\min},p_{\max}]$.
    Define a function~$f$ in the following fashion: when $M$ is even
    \begin{equation}\label{eqn:int_max_func_even}
        f(p,M)\coloneqq 
        \begin{cases}(1-\gamma) \cdot g_e\!\left(\frac{p(1)}{\gamma}\right) & \frac{p(1)}{\gamma} \in\left[0, p^e_{\max}\right] \\ \gamma-1 & \frac{p(1)}{\gamma} \in\left[p^e_{\max}, 1\right]\end{cases},
    \end{equation}
    and when $M$ is odd, 
    \begin{equation}\label{eqn:int_max_func_odd}
        f(p,M)\coloneqq 
        \begin{cases}(1-\gamma) \cdot g_o\!\left(\frac{p(1)}{\gamma}\right) & \frac{p(1)}{\gamma} \in\left[0, p^o_{\max}\right] \\ \gamma-1 & \frac{p(1)}{\gamma} \in\left[p^o_{\max}, 1\right]\end{cases},
    \end{equation}
    From Theorems~\ref{thm:max_func_even} and~\ref{thm:max_func_odd}, we know that any choice of $f_{\max}(p,M)$ that is differentiable and satisfies $f_{\max}(p,M)\geq f(p,M)$ for all $p$ will be a valid max-tradeoff function for our EAT channels. Since the derivatives of $f$ in~\eqref{eqn:int_max_func_even} and~\eqref{eqn:int_max_func_odd} at $p(1)/\gamma=p^e_{\max}$ and $p(1)/\gamma=p^o_{\max}$, respectively, are infinite, we now choose a function $f_{\max}(p,M)$ such that $\left\|\nabla f_{\max}(p,M)\right\|_{\infty}$ is finite. Let
    \begin{equation}\label{eqn:f_max}
        f_{\max}(p,p_t,M) \coloneqq  \begin{cases}f (p,M) & p(1)\leq p_t(1) \\ a(p_t)\cdot p(1) +b(p_t) & p(1)> p_t(1)\end{cases},
    \end{equation}
    where $p_t$ is a probability distribution over $\{0,1,\perp\}$, 
    \begin{align}
        a(p_t) & \coloneqq \left. \frac{\partial}{\partial p(1)}f(p,M)\right|_{p(1)=p_t(1)}, \text{   and}\\
        b(p_t) & \coloneqq f(p_t)-a(p_t)\cdot p_t(1).
    \end{align}
    
    It follows from the definition of $f_{\max}(p,M)$ that $f_{\max}(p,M)$ is differentiable and, for all $p_t$, $\left\|\nabla f_{\max}(\cdot,p_t,M)\right\|_{\infty}\leq a(p_t)$. Furthermore, as $f$ is a concave function, $f_{\max}(p,p_t,M)\geq f(p,M)$ for all $p$. Thus, $f_{\max}(p,p_t,M)$ is a max-tradeoff function. 
 
\subsection{Applying the EAT}

    We can finally apply the EAT, stated in Theorem~\ref{thm:eat}, to derive an upper bound on the conditional smooth max-entropy. By obtaining the exact form of $f_{\max}$, we have the calculated $t$ on the  right-hand side of~\eqref{eqn:h_max_eat}. We need to calculate~\eqref{eqn:h_max_eat_1}. We know that the dimension of all quantum registers involved is two, which means that $d_{O_{i}} =2M^\prime$. From~\eqref{eqn:f_max}, we get $\left\| \triangledown f_{\max}\right\|_{\infty} \leq a(p_t)$. According to Condition~\ref{cond:soundness} in Definition~\ref{def:dimec_protocol}, for soundness, we require that, for every source $\phi$ and measurement device, Protocol~\ref{pro:ent_test} either certifies a minimum amount of GHZ entanglement or  aborts with probability greater than $1-\varepsilon_{\operatorname{snd}}$ when applied on~$\sigma$. So, the probability of Protocol~\ref{pro:ent_test} not aborting is negligible i.e., $\Pr[\Omega]_{\tau} \leq \varepsilon_{\operatorname{snd}}$. Also, note that $\varepsilon_{\operatorname{smo}}$ is chosen such that the fidelity constraint in~\eqref{eqn:one_shot_distillable_ent} is satisfied. When Protocol~\ref{pro:ent_test_2} does not abort, we can define the following quantity, which corresponds to the right-hand side of~\eqref{eqn:h_max_eat}:
    \begin{multline}
        \eta(\omega_{\mathrm{exp}}\gamma - \delta_{\mathrm{est}}, p_t, \varepsilon_{\mathrm{smo}}, \varepsilon_{\operatorname{snd}},M^\prime,M)\coloneqq \\ n\cdot f_{\max }\!\left(\omega_{\mathrm{exp}}\gamma - \delta_{\mathrm{est}}, p_t,M\right)+
        2 \sqrt{n} \left(\log_2(1+2M^\prime)+\lceil a(p_t) \rceil \right)  \\
        \times \sqrt{1-2 \log_2 \left(\varepsilon_{\mathrm{smo}} \cdot \varepsilon_{\operatorname{snd}}\right)},
    \end{multline}
    where $\omega_{\mathrm{exp}}\gamma - \delta_{\mathrm{est}}$ is the minimum acceptable $p(1)$ as outlined in Step~\ref{prostep:abort_ET_2} of Protocol~\ref{pro:ent_test_2}. We still have one free variable $p_t$. We can optimize over $p_t$ to get the following equation.
    \begin{multline}\label{eqn:final_eta}
        \eta_{\mathrm{opt}}\left(\omega_{\mathrm{exp}}\gamma - \delta_{\mathrm{est}},\varepsilon_{\mathrm{smo}}, \varepsilon_{\operatorname{snd}},M^\prime,M\right) \coloneqq   
        \\
        \min _{p_t: p_{\min}<\frac{p_t(1)}{\gamma}<p_{\max}} \eta\left(\omega_{\mathrm{exp}}\gamma - \delta_{\mathrm{est}}, p_t, \varepsilon_{\mathrm{smo}}, \varepsilon_{\operatorname{snd}},M^\prime,M\right) .
    \end{multline}
    Finally, we can state the following:
    
    \begin{theorem}\label{thm:penult_thm_h_max}
        For every source and all measurement devices in the setting detailed earlier, let $\rho$ be the state generated using Protocol~\ref{pro:ent_test_2}, $\Omega$  the event that Protocol~\ref{pro:ent_test_2} does not abort, and $\rho_{|\Omega}$ the state conditioned on $\Omega$. Then, for all  $8\cdot3^{M/2}\sqrt{\varepsilon_{\operatorname{smo}}}, \varepsilon_{\operatorname{snd}}\in (0,1)$, either Protocol~\ref{pro:ent_test_2} aborts with probability greater than $1-\varepsilon_{\operatorname{snd}}$ or
        \begin{multline}
            H^{\varepsilon_{\mathrm{smo}}}_{\max}\!\left(\hat{A}_{[M^\prime],[n]} \middle\vert \hat{A}_{[M^\prime+1,M],[n]}\,(ADX)_{[M],[n]}\,T\right)_{\rho_{|\Omega}} \\ < \eta_{\mathrm{opt}}\!\left(\omega_{\mathrm{exp}},\varepsilon_{\mathrm{smo}}, \varepsilon_{\operatorname{snd}},M^\prime,M\right),
        \end{multline}
        where $\eta_{\mathrm{opt}}$ is defined in~\eqref{eqn:final_eta}.
    \end{theorem}
    
    \begin{proof}
        In Subsection~\ref{subsec:eat_channels}, we showed that Protocol~\ref{pro:ent_test_2} consists of EAT channels (defined in Definition~\ref{def:eat_channels}). In Subsection~\ref{subsec:max_tradeoff}, we derived the appropriate max-tradeoff function (as defined in Definition~\ref{def:max_tradeoff_func}). Hence, using Theorem~\ref{thm:eat}, the claim follows. 
    \end{proof}
    \medskip

    Now, we have all the ingredients to show that Protocol~\ref{pro:ent_test} is sound.
    
\subsection{Soundness}
    
    In this subsection, we show that Protocol~\ref{pro:ent_test} is sound. To do this, we need to bound the quantity in \eqref{eqn:one_shot_lower_bound} in a device-independent manner. We use the preceding theorems to obtain this bound and certify the multipartite entanglement distillation rate from a source. This brings us to our final theorem, which is as follows.
    \begin{theorem}
    \label{thm:main-result}
    Fix $\varepsilon_{\operatorname{smo}} > 0 $ such that $\varepsilon \coloneqq 8\cdot3^{M/2}\sqrt{\varepsilon_{\operatorname{smo}}}\in (0,1)$.
        For all $\varepsilon_{\operatorname{snd}},\varepsilon_{\mathrm{cmp}}, \varepsilon\in (0,1)$, Protocol~\ref{pro:ent_test} is a DIMEC protocol with
        \begin{enumerate}
            \item Noise tolerance (completeness): The probability that Protocol~\ref{pro:ent_test} aborts when applied on $\phi^{\operatorname{honest}}\in\mathcal{S}^{\operatorname{honest}}$ using a measurement device from $\mathcal{D}^{\operatorname{honest}}$ is at most $\varepsilon_{\operatorname{cmp}}\leq \exp\!\left(-2n\delta^{2}_{\operatorname{est}}\right)$, as shown in Theorem~\ref{thm:completeness}.
    
            \item Entanglement certification (soundness): For every source $\sigma$ and measurement device, either Protocol~\ref{pro:ent_test} aborts with probability greater than $1-\varepsilon_{\operatorname{snd}}$ when applied on $\sigma$ or 
            \begin{multline}
            E_{D}^{\varepsilon}(\rho_{\vert\Omega})\geq \\\frac{-\eta_{\mathrm{opt}}\left(\omega_{\mathrm{exp}},\varepsilon_{\operatorname{snd}}, \varepsilon_{\mathrm{smo}},M-1,M\right)}{M-1} +\frac{2\log\varepsilon_{\operatorname{smo}}}{M-1},
            \end{multline}
            where $\eta_{\mathrm{opt}}$ is defined in~\eqref{eqn:final_eta}.
        \end{enumerate}
    \end{theorem}
    
    \begin{proof}
        Using Theorem~\ref{thm:penult_thm_h_max}, we get
        \begin{multline}
            -H^{\varepsilon_{\mathrm{smo}}}_{\max}\!\left(\hat{A}_{[M^\prime],[n]} \middle\vert \hat{A}_{[M^\prime+1,M],[n]}\,(ADX)_{[M],[n]}\,T\right)_{\rho_{|\Omega}} \\ \geq -\eta_{\mathrm{opt}}\left(\omega_{\mathrm{exp}},\varepsilon_{\mathrm{smo}}, \varepsilon_{\operatorname{snd}},M^\prime,M\right),
        \end{multline}
        Notice that all our prior analysis is agnostic to what $i\in[M]$ is assigned to which specific party involved. Also recall that $M^\prime\in[M-1]$. Hence, 
        \begin{multline}
            -H^{\varepsilon_{\mathrm{smo}}}_{\max}\!\left(\hat{A}_{K,[n]} \middle\vert \hat{A}_{[M]\backslash K,[n]}\,(ADX)_{[M],[n]}\,T\right)_{\rho_{|\Omega}} \\ \geq -\eta_{\mathrm{opt}}\left(\omega_{\mathrm{exp}},\varepsilon_{\mathrm{smo}}, \varepsilon_{\operatorname{snd}},M^\prime,M\right),
        \end{multline}
        where $K\subseteq [M]\backslash k$ for some $k\in [M]$. 
        
        This bound holds regardless of the individual elements that constitute $K$. We can then use the above inequality to obtain a lower bound on the quantity in~\eqref{eqn:one_shot_lower_bound}. Hence, we get
        \begin{widetext}
        \begin{multline}
        \max_{k\in [M]}\left\{\min_{K\subseteq [M]\backslash k}\frac{-H^{\varepsilon_{\mathrm{smo}}}_{\max}\!\left(\hat{A}_{K,[n]} \middle\vert \hat{A}_{[M]\backslash K,[n]}\,(ADX)_{[M],[n]}\,T\right)_{\rho_{|\Omega}}}{|K|}\right\}+\frac{2\log\varepsilon_{\operatorname{smo}}}{M-1} \\\geq \max_{k\in [M]}\left\{\min_{K\subseteq [M]\backslash k}\frac{-\eta_{\mathrm{opt}}\left(\omega_{\mathrm{exp}},\varepsilon_{\operatorname{snd}}, \varepsilon_{\mathrm{smo}},|K|,M\right)}{|K|}\right\} + \frac{2\log\varepsilon_{\operatorname{smo}}}{M-1}.
        \end{multline}
        \end{widetext}
        The optimizations involved can be solved since $\eta_{\operatorname{opt}}(\varepsilon_{\operatorname{snd}}, \varepsilon_{\mathrm{smo}},|K|,M)$ depends only on $|K|$, leading to
        \begin{multline}
        E_{D}^{\varepsilon}(\rho_{\vert\Omega})\\\geq \frac{-\eta_{\mathrm{opt}}\left(\omega_{\mathrm{exp}},\varepsilon_{\operatorname{snd}}, \varepsilon_{\mathrm{smo}},M-1,M\right)}{M-1} + \frac{2\log\varepsilon_{\operatorname{smo}}}{M-1},
        \end{multline}
        which concludes the proof.
    \end{proof}
    \medskip
    
    We plot also $-\eta_{\mathrm{opt}}(\omega_{\exp})/n(M-1)$ for $M=4$ and $\gamma =0.5$ up to leading order in Figure~\ref{fig:plot_even}.
    \begin{figure}[h]
        \centering
        \includegraphics[width=\linewidth]{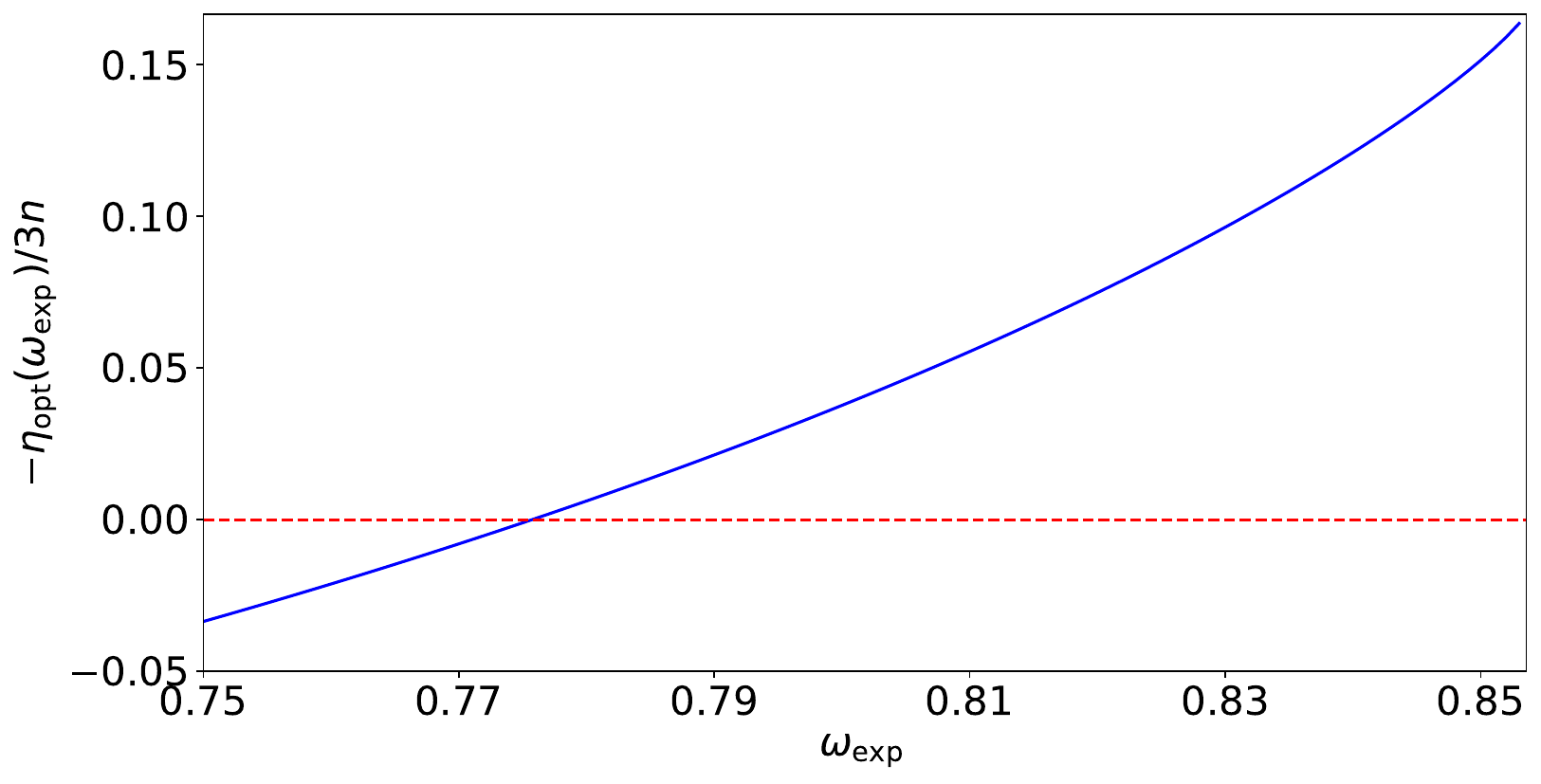}
        \caption{Plot of $-\eta_{\mathrm{opt}}(\omega_{\exp})/n(M-1)$, defined in~\eqref{eqn:final_eta}, for $M=4$ and $\gamma =0.5$ up to leading order in $n$ for $\omega_{\exp}\in[p^e_{\min},p^e_{\max}]$ defined in~\eqref{eqn:p_win_even}. The solid line corresponds to $-\eta_{\mathrm{opt}}(\omega_{\exp})/n(M-1)$. The dashed line corresponds to $-\eta_{\mathrm{opt}}/n(M-1)=0$. The quantity $-\eta_{\mathrm{opt}}(\omega_{\exp})/n(M-1)$ represents the amount of multipartite distillable entanglement that can be certified per round by Protocol~\ref{pro:ent_test} as function of $\omega_{\exp}$.}
        \label{fig:plot_even}
    \end{figure}
    
\section{Conclusion}

    In this work, we defined a DI, $M$-partite entanglement distillation certification protocol consisting of completeness and soundness conditions. We then proposed a protocol and showed that it is indeed a DI multipartite entanglement distillation certification protocol. Specifically, we proved that Protocol~\ref{pro:ent_test} is indeed a DI multipartite entanglement distillation certification protocol. Our proposed protocol is centered around the MABK inequality~\cite{PhysRevA.46.5375,PhysRevLett.65.1838,AVBelinskiĭ_1993}, which is critical to proving the completeness and soundness conditions. To prove the soundness condition, we used the entropy accumulation theorem~\cite{dupuis2016entropy}, the structure of the MABK~inequality~\cite{PhysRevA.46.5375,PhysRevLett.65.1838,AVBelinskiĭ_1993}, and the multipartite entanglement distillation protocols and rates described in~\cite{colomer2023decoupling,salek2023new}.

    For future works, we are interested in using other Bell-type inequalities that involve more than two binary measurements or measurements that have more than two outcomes. In this work, we have focused entirely on GHZ states and leaving out the W state. The W and GHZ states are not inter-convertible using LOCC alone~\cite{PhysRevA.62.062314}. Hence, our analysis here for GHZ states does not immediately apply to W states, but it is an interesting extension. It is also worth noting that there are infinitely many classes of genuinely multipartite entangled states that involve four or more parties~\cite{PhysRevLett.111.110502}. It will be interesting to see what a unified certification protocol for such classes of states might look like.

    We also note that our results are centered around a protocol based on state merging~\cite{colomer2023decoupling,salek2023new}. One could always consider another a protocol that does not involve state merging. An approach that does not rely on the equivalence between Bell inequality and MABK inequality may be necessary for this scenario.

\begin{acknowledgments}
    We thank Ashutosh Marwah for pointing us to \cite[Appendix~A]{9996821}. Part of this work was completed during the conference “Beyond IID in Information Theory,” held at the University of Illinois Urbana-Champaign from July 29 to August 2, 2024, and supported by NSF Grant No. 2409823. We thank Kaiyuan Ji, Dhrumil Patel, Vishal Singh, and Theshani Nuradha for helpful discussions. We acknowledge funding from Air Force Office of Scientific Research Award Nos.~FA9550-20-1-0067 and FA8750-23-2-0031. We also acknowledge support from the School of Electrical and Computer Engineering, Cornell University.

    This material is based on research sponsored by Air Force Research Laboratory under agreement number FA8750-23-2-0031. The U.S.~Government is authorized to reproduce and distribute reprints for Governmental purposes, notwithstanding any copyright notation thereon. The views and conclusions contained herein are those of the authors and should not be interpreted as necessarily representing the official policies or endorsements, either expressed or implied, of Air Force Research Laboratory or the U.S.~Government. 
\end{acknowledgments}

\bibliography{ref}

\appendix
\onecolumngrid

\section{Proof of Lemma~\ref{thm:mabk_2_bell}}\label{proof:mabk_2_bell}

\begin{theorem}
     An $M$-partite MABK inequality~\eqref{eqn:inequality_mabk} is equivalent to a CHSH inequality within any bipartition consisting of $M^\prime$ parties on  side and the $M -M^\prime$ other parties on the other side such that, for some $\hat{O}^{[M-M^\prime+1]}_0$ and $\hat{O}^{[M-M^\prime+1]}_1$, 
    \begin{equation}
        \mathcal{K}_M=\frac{1}{2}\mathcal{F}\!\left(\mathcal{K}_{M-M^\prime},\overline{\mathcal{K}}_{M-M^\prime},\hat{O}^{[M-M^\prime+1]}_0, \hat{O}^{[M-M^\prime+1]}_1\right).
    \end{equation}
\end{theorem}
\begin{proof}
    We know that, for all $M\geq 3$
    \begin{equation}\label{eqn:m}
        \mathcal{K}_M = \frac{1}{2}\mathcal{K}_{M-1}\otimes\left(\hat{O}^M_0 + \hat{O}^M_1 \right) + \frac{1}{2}\overline{\mathcal{K}}_{M-1}\otimes\left(\hat{O}^M_0 - \hat{O}^M_1 \right),
    \end{equation}
    and 
    \begin{equation}
        \overline{\mathcal{K}}_M = \frac{1}{2}\overline{\mathcal{K}}_{M-1}\otimes\left(\hat{O}^M_0 + \hat{O}^M_1 \right) + \frac{1}{2 }\mathcal{K}_{M-1}\otimes\left(\hat{O}^M_1 - \hat{O}^M_0 \right).
    \end{equation}
    It is then clear that
    \begin{equation}\label{eqn:m-1}
        \mathcal{K}_{M-1} = \frac{1}{2 }\mathcal{K}_{M-2}\otimes\left(\hat{O}^{M-1}_0 + \hat{O}^{M-1}_1 \right) + \frac{1}{2 }\overline{\mathcal{K}}_{M-2}\otimes\left(\hat{O}^{M-1}_0 - \hat{O}^{M-1}_1 \right),
    \end{equation}
    Let 
    \begin{equation}
        V_M \coloneqq \left|\operatorname{Tr}\!\left[\mathcal{K}_M\, \rho_{\hat{A}_{[M]}}\right]\right| = \frac{1}{2}\left|\operatorname{Tr}\!\left[\left\{\mathcal{K}_{M-1}\otimes\left(\hat{O}^M_0 + \hat{O}^M_1 \right) + \overline{\mathcal{K}}_{M-1}\otimes\left(\hat{O}^M_0 - \hat{O}^M_1 \right)\right\} \rho_{\hat{A}_{[M]}}\right]\right|.
    \end{equation}
    From \cite[Theorem 2]{ribeiro2017MABK}, we know be relabeling $\mathcal{K}_{M-1}$ as $\hat{O}_0^1$, $\overline{\mathcal{K}}_{M-1}$ as $\hat{O}_1^1$, $\hat{O}^{M}_0$ as $\hat{O}_0^2$,  and $\hat{O}^{M}_1 $ as $\hat{O}_1^2$, for $M^\prime = 1$
    \begin{align}
        V_M = \frac{1}{2}\left|\operatorname{Tr}\!\left[\mathcal{F}\left(\mathcal{K}_{M-1},\overline{\mathcal{K}}_{M-1},\hat{O}^M_0, \hat{O}^M_1 \right) \rho_{\hat{A}_{[M]}}\right]\right|.
    \end{align}
    We need to show the same for $1\leq M^\prime\leq M-1$. Let us start with $M^\prime=2$. By substituting~\eqref{eqn:m-1} in~\eqref{eqn:m}, we get
    \begin{align}
        \mathcal{K}_M &=   \frac{1}{4}\left(\mathcal{K}_{M-2}\otimes\left(\hat{O}^{M-1}_0 + \hat{O}^{M-1}_1 \right) + \overline{\mathcal{K}}_{M-2}\otimes\left(\hat{O}^{M-1}_0 - \hat{O}^{M-1}_1\right)\right) \otimes\left(\hat{O}^M_0 + \hat{O}^M_1 \right) \notag\\
        &\qquad + \frac{1}{4}\left(\overline{\mathcal{K}}_{M-2}\otimes\left(\hat{O}^{M-1}_0 + \hat{O}^{M-1}_1 \right) + \mathcal{K}_{M-2}\otimes\left(\hat{O}^{M-1}_1 - \hat{O}^{M-1}_0\right)\right) \otimes\left(\hat{O}^M_0 - \hat{O}^M_1 \right)\\
        &=   \frac{1}{4}\mathcal{K}_{M-2}\otimes\left(\left(\hat{O}^{M-1}_0 + \hat{O}^{M-1}_1 \right)\otimes\left(\hat{O}^{M}_0 + \hat{O}^{M}_1\right)+ \left(\hat{O}^{M-1}_1 - \hat{O}^{M-1}_0 \right)\otimes\left(\hat{O}^{M}_0 - \hat{O}^{M}_1\right)\right)\notag\\
        &\qquad+  \frac{1}{4}\overline{\mathcal{K}}_{M-2}\otimes\left(\left(\hat{O}^{M-1}_0 + \hat{O}^{M-1}_1 \right)\otimes\left(\hat{O}^{M}_0 - \hat{O}^{M}_1\right)+ \left(\hat{O}^{M-1}_1 - \hat{O}^{M-1}_0 \right)\otimes\left(\hat{O}^{M}_0 + \hat{O}^{M}_1\right)\right)\\
        &=  \frac{1}{2}\mathcal{K}_{M-2}\otimes\left(\hat{O}^{M-1}_0 \otimes \hat{O}^{M}_1 + \hat{O}^{M-1}_1 \otimes \hat{O}^{M}_0\right) +  \frac{1}{2}\overline{\mathcal{K}}_{M-2}\otimes\left(\hat{O}^{M-1}_1 \otimes \hat{O}^{M}_0 - \hat{O}^{M-1}_0 \otimes \hat{O}^{M}_1\right).
    \end{align}
    Relabelling $\hat{O}^{M-1}_1 \otimes \hat{O}^{M}_0$ as $\hat{O}_0^{[M-1]}$,  and $\hat{O}^{M-1}_0 \otimes \hat{O}^{M}_1 $ as $\hat{O}_1^{[M-1]}$, we get 
    \begin{align}\label{eqn:M_prime_2}
        V_M = \frac{1}{2}\left|\operatorname{Tr}\!\left[\mathcal{F}\left(\mathcal{K}_{M-2},\overline{\mathcal{K}}_{M-2},\hat{O}_0^{[M-1]}, \hat{O}_0^{[M-1]} \right) \rho_{\hat{A}_{[M]}}\right]\right|.
    \end{align}
    Now, we have shown that the proposition is true for $M^\prime=1$ and $M^\prime=2$. We can now use induction to show that our proposition holds for all $1\leq M^\prime\leq M-1$. Let assume that for some $M^\prime = m$, the following hold:
    \begin{equation}\label{eqn:M-m}
        \mathcal{K}_M = \frac{1}{2}\mathcal{K}_{M-m}\otimes\left(\hat{O}^{[M-m+1]}_0 + \hat{O}^{[M-m+1]}_1 \right) + \frac{1}{2}\overline{\mathcal{K}}_{M-m}\otimes\left(\hat{O}^{[M-m+1]}_0 - \hat{O}^{[M-m+1]}_1 \right)
    \end{equation} 
    and 
    \begin{align}
        V_M = \frac{1}{2}\left|\operatorname{Tr}\!\left[\mathcal{F}\left(\mathcal{K}_{M-m},\overline{\mathcal{K}}_{M-m}, \hat{O}^{[M-m+1]}_0,\hat{O}^{[M-m+1]}_1\right) \rho_{\hat{A}_{[M]}}\right]\right|.
    \end{align}
    where, $\mathcal{K}_{M-m}$ is defined in~\eqref{eqn:def_mabk_operator}, and $\hat{O}^{[M-m+1]}_0$ and $\hat{O}^{[M-m+1]}_1$ are observables similar to those in~\eqref{eqn:M_prime_2} involving parties $M-m+1,\ldots,M$. We know that 
    \begin{equation}\label{eqn:M-m-1}
        \mathcal{K}_{M-m} = \frac{1}{2}\mathcal{K}_{M-m-1}\otimes\left(\hat{O}^{M-m}_0 + \hat{O}^{M-m}_1 \right) + \frac{1}{2}\overline{\mathcal{K}}_{M-m}\otimes\left(\hat{O}^{M-m}_0 - \hat{O}^{M-m}_1 \right).
    \end{equation} 
    By substituting~\eqref{eqn:M-m-1} in~\eqref{eqn:M-m}, we get
    \begin{align}
        \mathcal{K}_M&= \frac{1}{4}\mathcal{K}_{M-(m+1)}\otimes\left(\left(\hat{O}^{M-m}_0 + \hat{O}^{M-m}_1 \right)\otimes\left(\hat{O}^{[M-m+1]}_0 + \hat{O}^{[M-m+1]}_1 \right)\right)\notag\\
        &\qquad \qquad+ \frac{1}{4}\mathcal{K}_{M-(m+1)}\otimes\left(\left(\hat{O}^{M-m}_0 - \hat{O}^{M-m}_1 \right)\otimes\left(\hat{O}^{[M-m+1]}_0 - \hat{O}^{[M-m+1]}_1 \right)\right)\notag\\
        &\qquad \qquad \qquad \qquad+ \frac{1}{4}\overline{\mathcal{K}}_{M-(m+1)}\otimes\left(\left(\hat{O}^{M-m}_0 + \hat{O}^{M-m}_1 \right)\otimes\left(\hat{O}^{[M-m+1]}_0 - \hat{O}^{[M-m+1]}_1 \right)\right)\notag\\ 
        &\qquad \qquad\qquad \qquad\qquad \qquad+ \frac{1}{4}\overline{\mathcal{K}}_{M-(m+1)}\otimes\left(\left(\hat{O}^{M-m}_0 - \hat{O}^{M-m}_1 \right)\otimes\left(\hat{O}^{[M-m+1]}_0 + \hat{O}^{[M-m+1]}_1 \right)\right)\\
        &=  \frac{1}{2}\mathcal{K}_{M-(m+1)}\otimes\left(\hat{O}^{M-m}_0 \otimes \hat{O}^{[M-m+1]}_1 + \hat{O}^{M-m}_1 \otimes \hat{O}^{[M-m+1]}_0\right) \notag\\
        &\qquad \qquad\qquad \qquad\qquad \qquad  +\frac{1}{2}\overline{\mathcal{K}}_{M-(m+1)}\otimes\left(\hat{O}^{M-m}_1 \otimes \hat{O}^{[M-m+1]}_0 - \hat{O}^{M-m}_0 \otimes \hat{O}^{[M-m+1]}_1\right).
    \end{align}
    Therefore,
    \begin{align}
        V_M = \frac{1}{2}\left|\operatorname{Tr}\!\left[\mathcal{F}\left(\mathcal{K}_{M-(m+1)},\overline{\mathcal{K}}_{M-(m+1)}, \hat{O}^{M-m}_1 \otimes \hat{O}^{[M-m+1]}_0, \hat{O}^{M-m}_0 \otimes \hat{O}^{[M-m+1]}_1\right) \rho_{\hat{A}_{[M]}}\right]\right|.
    \end{align}
    Hence, by strong induction, we have proven our proposition.
\end{proof}

\end{document}